\documentclass[journal]{IEEEtran}
\ifCLASSINFOpdf
\else
\fi
\usepackage{amsmath,amssymb,amsthm}
\usepackage{authblk}
\newtheorem{theorem}{Theorem}
\usepackage{hhline}
\usepackage{graphicx}
\usepackage{epstopdf}
\usepackage{subcaption}
\usepackage{algorithm}
\usepackage{algorithmic}
\usepackage{amsmath}
\usepackage{cite}

\begin{document}

\title{On The Continuous Coverage Problem for a Swarm of UAVs}

\author[1]{Hazim Shakhatreh}
\author[1]{Abdallah Khreishah}
\author[2]{Jacob Chakareski}
\author[3]{Haythem Bany Salameh}
\author[4]{Issa Khalil}
\affil[1]{Department of Electrical and Computer Engineering, New Jersey Institute of Technology}
 \affil[2]{Department of Electrical and Computer Engineering, University of Alabama}
\affil[3]{Department of Telecommunications Engineering, Yarmouk University}
\affil[4]{Qatar Computing Research Institute}

\maketitle

\begin{abstract}
Unmanned aerial vehicles (UAVs) can be used to provide wireless network and remote surveillance coverage for disaster-affected areas. During such a situation, the UAVs need to return periodically to a charging station for recharging, due to their limited battery capacity. We study the problem of minimizing the number of UAVs required for a continuous coverage of a given area, given the recharging requirement. We prove that this problem is NP-complete. Due to its intractability, we study partitioning the coverage graph into cycles that start at the charging station. We first characterize the minimum number of UAVs to cover such a cycle based on the charging time, the traveling time, and the number of subareas to be covered by the cycle. Based on this analysis, we then develop an efficient algorithm, the cycles with limited energy algorithm. The straightforward method to
continuously cover a given area is to split it into N subareas
and cover it by N cycles using N additional UAVs.
 Our simulation results examine the importance of critical system parameters: the energy capacity of the UAVs, the number of subareas in the covered area, and the UAV charging and traveling times. We demonstrate that the cycles with limited energy algorithm requires 69\%-94\% fewer additional UAVs relative to the straightforward method, as the energy capacity of the UAVs is increased, and 67\%-71\% fewer additional UAVs, as the number of subareas is increased.
\end{abstract}


\section{Introduction}
\label{sec:Introduction}
In 2005, Hurricane Katrina in the United States caused over 1,900 deaths, 3 million land-line phone interruptions, and more than 2,000 base stations going out of service~\cite{bupe2015relief,miller2006hurricane,townsend2006federal}. Another example of a large-scale interruption of telecommunications service is the World Trade Center attack in 2001, when it took just minutes for the nearby base stations to be overloaded. The attacks caused the disturbance of a phone switch with over 200,000 lines, 20 cell sites, and 9 TV broadcast stations~\cite{bupe2015relief,noam2004world}. These incidents demonstrate the need for a quick/efficient deployment network for emergency cases. 

The authors in ~\cite{dalmasso2012wimax} proposed a UAV-based replacement network during disasters, where the UAVs serve as aerial wireless base stations. However, this study did not consider how the UAVs will guarantee a continuous coverage when they need to return to the charging station for recharging. Though a UAV has limited energy capacity and needs to recharge its battery before running out of power during the coverage process, only few studies have considered this constraint in the UAV coverage problem. Concretely, the author in~\cite{ha2010uav} determined the minimum number of UAVs that can provide continuous 
	\begin{figure}[!t]
		\centering
		\includegraphics[scale=0.26]{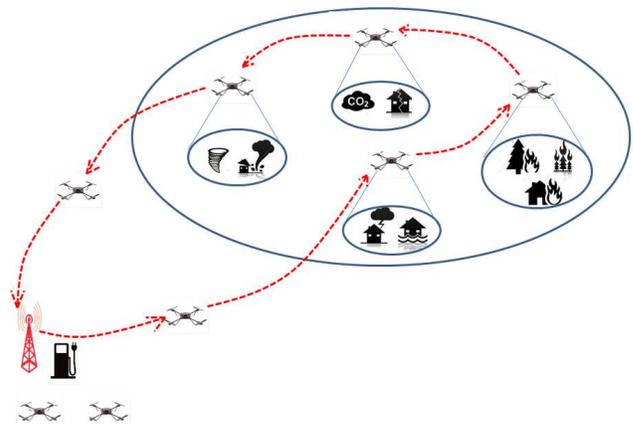}
		\caption{The Continuous Coverage Problem}
		\label{fig:system}
	\end{figure}
coverage for a single area using UAVs with uniform and non-uniform energy capacity. However, no consideration has been made for the case when there are multiple subareas that need to be covered, which is the typical scenario during disasters. The authors in \cite{wu2014collaborative,zhang2015collaborative} formulated the Mobile Charging Problem, in which multiple mobile chargers collaborate to charge static sensors with minimum number of mobile chargers subject to speed and energy limits of the mobile chargers, such that the energy consumption of the mobile chargers is minimized. Note that the mobile charging problem is different from the problem we study. In the mobile charging problem, the chargers will not cover the sensors continuously. The mobile charger will visit the sensor and stay for a specific time to charge the sensor and after finishing the charging process, it will visit other sensors. 

In~\cite{pugliese2015modelling}, the authors studied the continuous coverage problem for mobile targets, where during the coverage process a UAV that runs out of energy is replaced by a new one. Many studies~\cite{mozaffari2015drone,mozaffari2016optimal,zorbas2013energy,ceran2015optimal} focused on minimizing the total transmission power of the UAVs during the coverage of a geographical area, however, no limits on the UAV energy capacity and the need for recharging have been considered. The work in \cite{guptasurvey} reported that the energy consumption during data transmission and reception is much smaller than the energy consumption during the UAV hovering, i.e., it only constitutes 10$\%$-20$\%$ of the UAV energy capacity. Thus, it is important to conduct studies that take into account the energy consumption during the UAV hovering rather than focusing on minimizing the energy consumption during data transmission and reception.

Contrary to the related work above, we integrate the recharging requirements into the coverage problem and examine the minimum number of required UAVs for enabling continuous coverage under that setting (see Figure~\ref{fig:system}). To the best of our knowledge, this is the first study that jointly considers the coverage and recharging problems
where multiple subareas are to be covered. Our main contributions in this context are: 1) We formulate the problem of minimizing the number of UAVs required to provide continuous coverage of a given area, given the recharging requirement. 2) We prove that this problem is NP-complete. 3) Due to the intractability of the problem, we study partitioning the coverage graph into cycles that start at the charging station. 4) Based on this analysis, we develop an efficient algorithm.

The rest of this paper is organized as follows. Section~\ref{sec:system_model} presents our system model including the problem formulation and the proof of the NP completeness. In Section~\ref{sec:heurisitic}, we show how to find the minimum number of additional UAVs that are required to guarantee the continuous coverage. Then, we present our proposed algorithm. In Section~\ref{sec:Experiments}, we present our experimental results. Finally, Section~\ref{sec:Conclusion} concludes the paper.

\section{System Model}
\label{sec:system_model}
\subsection{Problem Statement}
\label{subsec:problem statement}

Consider a geographical area $G$=\{$g_{1}$,...,$g_{N}$\}, where $g_{i}$ represents a subarea $i$, the subarea $g_{1}$ $\in$ $G$ includes the charging station and all subareas except subarea $g_{1}$ need to be covered $G_{0}$=$G$ $\backslash$ $g_{1}$. We want to find the minimum number of UAVs that can provide a continuous coverage over $G_{0}$ by placing the UAVs at locations where each UAV will provide full coverage for one subarea. In the continuous coverage problem, we assume: (1) Time slotted system in which the slot duration is 1 time unit and the total coverage duration is $T$. (2) All UAVs start the coverage process from the charging station and they need to return to the charging station after they complete the coverage process. (3) Each UAV has limited energy capacity $E$ and it needs to return to the charging station to recharge the battery before running out during the coverage process. (4) Each UAV can move (from the charging station to location $i$), (from location $i$ to location $j$) or (from location $j$ to the charging station) and this process will take one time slot. (5) Each UAV covers a given subarea for one or multiple time slots. (6) Each subarea will be covered by only one UAV. (7) The UAV cannot travel to the charging station or to any other location until the handoff process is completed in which another UAV arrives to cover the subarea such that the continuous coverage is guaranteed. (8) The recharging process takes $T_{charge}$ at the charging station.

\subsection{Problem Formulation}
\label{subsec:problem_formlation}
Now, we formulate the continuous coverage problem. In order to present the problem formulation, we introduce the  binary variable $x_{m}$ that takes the value of 1 if the UAV $m$ visits any subarea from charging station during the coverage duration $T$ and equals 0 otherwise; the binary variable $y_{ij,m}^t$ that takes the value of 1 if the if the UAV $m$ moves through edge $ij$ during the time slot $t$ and equals 0 otherwise; the binary variable $z_{j,m}^t$ that takes the value 1 if the UAV $m$ covers the subarea $j$ at time slot $t$ and equals 0 otherwise. Table I provides a list of the major notations used in this paper.

\begin{table}[!h]
	\scriptsize
	\renewcommand{\arraystretch}{1.5}
	\caption{A List of Notations}
	\label{table}
	\centering
	\begin{tabular}{|c|l|}
		\hline
		$M$ & The set of fully charged UAVs available at the charging\\& station.\\
		\hline
		$E$ & The energy capacity of each UAV $m$ $\in$ $M$.\\
		\hline 
		$T$ & The coverage duration.\\
		\hline
		$E_{ij}^{Travel}$& The energy consumed when the UAV travels from subarea
		\\
		&$i$ to subarea $j$ where $i,j\in G$.\\
		\hline
		$E_{j}^{Cover}$ & The energy consumed when the UAV covers the subarea $j$ \\
		&for one time slot where $j\in G$ (constant).\\
		\hline
		$T_{charge}$ & The time that the UAV needs to recharge the battery at\\ &the charging station.\\
		\hline
	\end{tabular}
\end{table}
\begin{equation*}
	\begin{array}{ll}
		\displaystyle \min \displaystyle \sum_{m\in M}^{ } x_m  \\
		\textrm{subject to}\\
		y_{ij,m}^t  \leq x_m ~~~~~~ \forall i,j \in G,\forall t \in [0,T],
		\forall m\in M&(1) \\
		
		z_{j,m}^t \leq  x_m ~~~~~~~ \forall j \in G_{0},\forall t \in (0,T),\forall m\in M &(2)\\
		
		\displaystyle \sum_{m\in M}^{ } y_{1j,m}^0 = 1 ~~ \forall j \in G_{0}&(3)\\
		
		\displaystyle \sum_{m\in M}^{ } z_{j,m}^t =  1 ~~~ \forall j \in G_{0},\forall t \in (0,T)&(4)\\
		
		\displaystyle \sum_{i\in G,i\neq j} \displaystyle \sum_{m\in M}^{ } y_{ij,m}^t  \leq  1 ~~~~~~ \forall j \in G_{0},\forall t \in [0,T]&(5)\\
		
		\displaystyle \sum_{i_{1}\in G}^{ } y_{i_{1}j,m_{1}}^t  =  \displaystyle \sum_{i_{2}\in G}^{ } y_{ji_{2},m_{2}}^{t+1}  ~ \forall j \in G_{0},\\
		~~~~~~~~~~~~~~~~~~~~~~~~~~~~~~~~~\forall t \in (0,T),m_{1}\neq m_{2}&(6)\\
		
		\displaystyle \sum_{m\in M} \displaystyle \sum_{t\in [0,T)}^{ } \displaystyle \sum_{i\in G}^{ } y_{ij,m}^t \leq   \displaystyle \sum_{m\in M} \displaystyle \sum_{t\in (0,T)}^{ } z_{j,m}^t \\ 
		~~~~~~~~~~~~~~~~~~~~~~~~~~~~~~~~~\forall j \in G_{0}&(7)\\
		
		\displaystyle \sum_{j\in G_{0}}^{ }\displaystyle \sum_{\tau\in T_{charge}}^{ } [y_{j1,m}^t+y_{1j,m}^{t+\tau}]  \leq  1 \\~~~~~~~~~~~~~~~~~~~~~~~~~~~~~~~~ \forall m\in M,\forall t \in (0,T)&(8)\\
		
		\displaystyle \sum_{m\in M} \displaystyle \sum_{t\in [0,T]}^{ } \displaystyle \sum_{j\in G_{0}}^{ } y_{1j,m}^t  = \displaystyle \sum_{m\in M} \displaystyle \sum_{t\in [0,T]}^{ } \displaystyle \sum_{i\in G_{0}}^{ } y_{i1,m}^t&(9)\\
		
		\displaystyle \sum_{t\in [t_{1},t_{2}]}^{ } \displaystyle \sum_{i,j\in G}^{ } E_{ij}^{Travel} y_{ij,m}^t  +  \displaystyle \sum_{t\in [t_{1},t_{2}]}^{ } \displaystyle \sum_{j\in G}^{ } E_{j}^{Cover} z_{j,m}^t \\ \leq  
		E ~~~~~ \forall m\in M,\forall [t_{1},t_{2}]\in [0,T],t_{1}=arg~y_{1j,m}^t\\
		~~~~~~~~~~~t_{2}=arg~y_{i1,m}^t,t_{2} > t_{1}, \forall t_{3}\in(t_{1},t_{2}) \\
		~~~~~~~~~~~t_{3} \neq arg~y_{1j,m}^t \neq arg~y_{i1,m}^t &(10)
	\end{array}
\end{equation*}
 The objective is to minimize the number of UAVs that are needed to provide a continuous coverage during the coverage duration $T$ subject to various design constraints. Constraints (1) and (2) ensure that the UAV can travel and cover the subareas only if we select it to participate in the coverage process. Constraint (3) ensures that all subareas will be covered at the first time slot. Constraint (4) guarantees the continuous coverage for each subarea. Constraint (5) allows the UAV to visit a new subarea (when $y_{ij,m}^t$=1) or to continue covering the current subarea (when $y_{ij,m}^t$=0). Constraint (6) characterizes the handoff process between the UAVs, when the UAV $m_{1}$ wants to visit the subarea $j$ from subarea $i_{1}$ at time $t$ ($y_{i_{1}j,m_{1}}^t$=1), the UAV $m_{2}$ that covers the subarea $j$ will travel to subarea $i_{2}$ at time $t+1$ ($y_{ji_{2},m_{2}}^{t+1}$=1). Constraint (7) describes the relation between the traveling process and the covering process, where the number of times that the subarea $j$ is covered will be greater than or equal the number of times that it is visited. Constraint (8) shows that the recharging process will take $T_{charge}$ at the charging station. Constraint (9) ensures that the number of UAVs outgoing from the charging station and the number of UAVs incoming to charging station are the same after we complete the coverage process. Constraint (10) shows that the energy capacity of the UAV can cover the wasted energy during the traveling and the covering processes in each cycle where $t_{1}$ represents the time that the UAV travels from the charging station with full energy capacity and $t_{2}$ represents the time that the UAV arrives to the charging station to charge the battery. Now we will prove that the continuous coverage problem is an NP-complete.
   \vspace{-3mm}
\subsection{NP completeness}
\label{subsec:NP_complete}
\begin{theorem}
\label{theorem_one}	
 The Continuous Coverage Problem is NP-complete.
\end{theorem}
\begin{proof}
The number of constraints is polynomial in terms of the number of subareas, the number of UAVs and the number of time slots. Given any solution for our problem, we can check the
 solution\textquoteright s feasibility in polynomial time, then the problem is NP.\\
To prove that the problem is NP-hard, we reduce the Bin Packing Problem which is NP-hard~\cite{korf2002new} to a special case of our problem. The special case will be the discrete coverage problem in which each subarea will be visited one time by one UAV during the coverage process. In the Bin Packing Problem, we have $p$ items where each item has volume $z_p$. All items must be packed into a finite number of bins ($b_1$, $b$,...,$b_B$), each of volume $V$ in a way that minimizes the number of bins used. The reduction steps are: 1) The $b$-th bin in the Bin Packing Problem is mapped to the $m$-th UAV in our problem (where the volume $V$ for each bin is mapped to the energy capacity of the UAV $E$). 2) The $p$-th item is mapped to the $n$-th subarea, (where the volume for each item $p$ is mapped to the energy consumed when the UAV (visits and covers) subarea $n$. 3) All UAVs have the same energy capacity $E$. 4) The energy consumed (during the traveling and the covering processes) when the UAV visits subarea $j$ from any subarea ($i$ $\in$ $G$ $\backslash$ \{j\}) will be constant. 5) The energy required for the UAV to return to the charging station from any subarea $i$ will be zero ($E_{i1}^{Travel}$=0). 6) The time that the UAV needs to recharge the battery at the charging station will be infinity. 7) Each subarea will be visited one time by one UAV during the coverage process (discrete coverage). If there exists a solution to the bin packing problem with cost $C$, then the selected bins will represent the UAVs that are selected and the items in each bin will represent the subareas that the UAV must visit and the total cost of our problem is $C$. If there exists a solution to our problem with cost $C$, then the selected UAVs will represent the bins that are selected and the subareas that the UAV must visit will represent the items in each bin and the total cost of the bin packing problem is $C$. We prove that there exists a solution to the bin packing problem with cost $C$ iff there exists a solution to our problem with cost $C$.\qedhere
\end{proof}
\section{Heuristic algorithm} 
\label{sec:heurisitic}
Due to the intractability of the problem, we study partitioning the coverage graph into cycles that start at the charging station. We first characterize the minimum number of UAVs to cover each cycle based on the charging time, the traveling time, and the number of subareas to be covered by the cycle. Our analysis based on the uniform coverage in which the UAV covers each subarea in a given cycle for a constant time. Based on this analysis, we then develop an efficient algorithm, the cycles with limited energy algorithm, that minimizes the required number of UAVs that guarantees a continues coverage.
\vspace{-2mm}
\subsection{Analysis}
\label{subsec:analysis}
It is obvious that we need $N$ UAVs to cover $N$ subareas at  any given time, but the question here is how many additional UAVs are needed to guarantee a continuous coverage. In this subsection, we assume that the UAV visits the subareas based on a cycle that starts from the charging station and ends at the charging station for charging process. We also assume that a given UAV covers the subareas in the cycle uniformly, in which the UAV covers each subarea in a given cycle for a constant time.
 In Theorem~\ref{theorem_two}, we find the minimum number of additional UAVs that are needed to guarantee a continuous coverage for a cycle, which will help us while developing Algorithm 1.
\begin{theorem}
	\label{theorem_two}
The minimum number of additional UAVs $k$ that are required to provide continuous and uniform coverage for a cycle that contains $n$ subareas must satisfy this inequality:
\end{theorem}
\begin{equation}
k\frac{T_{Coverage}}{n} \geq (n+1)T + T_{Charge}
\end{equation}
where $T_{Coverage}$ is the time that the UAV allocates to cover all subareas in the cycle, $T$ is the time that the UAV needs to travel from subarea $i$ to subarea $j$ and $T_{charge}$ is the time that the UAV needs to recharge the battery at the charging station.
\begin{proof}
	Consider that all $n$ subareas in the cycle are covered by $n$ UAVs and the UAV that covers the last subarea want to return to the charging station to recharge its battery. The handoff process needs to begin between one of the additional UAVs from the charging station and the UAV that covers the first subarea in the cycle.\\
The UAV that covers the last subarea needs to wait $(n-1)$ $T$ to do the handoff process, during this time the additional UAVs are covering the first subarea. After the handoff process is completed, the UAV needs $T$ time units to return to the charging station, $T_{charge}$ to recharge the battery and T to visit the first subarea in the cycle again. Then, we have:

\begin{equation}
k\frac{T_{Coverage}}{n} \geq (n-1)T +T+ T_{Charge}+T
\end{equation}
\qedhere
\end{proof}
\vspace{-10mm}
\subsection{The cycles with limited energy algorithm}
\label{subsec:Algorithm}

The straightforward method (SM) to continuously cover $N$ subareas is to allocate two UAVs for each subarea. At the first time slot, $N$ UAVs cover the $N$ subareas. Then, any UAV wants to return to the charging station to recharge the battery will do the handoff process with one of the additional UAVs that are available at the charging station. By applying SM, we need $N$ additional UAVs and we have $N$ cycles to cover all the subareas.\\
Our proposed algorithm, the cycles with limited energy algorithm (CLE), is inspired by the nearest neighbor algorithm, the nearest neighbor algorithm is used to solve the Traveling Salesman Problem~\cite{johnson1997traveling}, in which the salesman keeps visiting the nearest unvisited vertex until all the vertices are visited. In our algorithm, the UAV (salesman) has limited energy capacity and before visiting any new subarea, we must check if the remaining energy is enough to return to the charging station from the new location or not. In the previous subsection, we show how to find the minimum number of additional UAVs that are required to guarantee the continuous coverage for a given cycle, we use the Theorem~\ref{theorem_two} to find the minimum number of additional UAVs that are required to provide the continuous coverage for a given area, by finding the cycles that need only one additional UAV. The pseudo code of this algorithm is shown in Algorithm 1.
	\begin{algorithm}
		\begin{algorithmic}
			\STATE \textbf{Input:}
			\STATE The geografical area $G$=\{$g_{1}$,...,$g_{N}$\},
			\STATE $T$: The required time to travel between two subareas in the area,
			\STATE $E$: The energy capacity of the UAV,
			\STATE $T_{Charge}$: The time that the UAV needs to recharge the battery at charging station,
			\STATE $e$: The energy consumed by the UAV when it covers the subarea for 1 sec.
			\STATE $i$=1     // $i$ is the index of the cycle
			\STATE \textbf{START:}
			\WHILE {$G$ not empty}
			\STATE $c_i$=\{$g_{1}$\}
			\STATE \textbf{Do:}
			\STATE 1 $v$= most recently added subarea to cycle $c_i$
			\STATE 2 Find \{$g$\}= $argmin_{b\in G-\{v\}}$ $distance(v,b)$
			\STATE 3 Calculate  $E_{Coverage}$=$E$-$E_{Travel}$-$E_{Return to BS}$
			\STATE 4 Calculate $T_{Coverage}$ $=$ $\frac{E_{Coverage}}{e}$
			\STATE 5 \textbf{If} $\frac{T_{Coverage}}{|c_i|}$ $\geq$ $(|c_i|+1)T$ $+$ $T_{Charge}$ \textbf{then}
			\STATE 6 $c_i$ $\longleftarrow$ $c_i$ $\cup$ $\{g\}$
			\STATE 7 $G$ $\longleftarrow$ $G$ $\backslash$ $\{g\}$
			\STATE \textbf{while} ($\frac{T_{Coverage}}{|c_i|}$ $\geq$ $(|c_i|+1)T$ $+$ $T_{Charge}$ )	
			\STATE 8 $c_i$ $\longleftarrow$ $c_i$ $\cup$ $\{g_{1}\}$
			\STATE 9  $\textbf{C}$ $\longleftarrow$ $\textbf{C}$ $\cup$ $c_i$
			\STATE 10 $i$=$i$+1
			\ENDWHILE
			\STATE \textbf{Output:} \textbf{C} and $M^{'}$ $=$( ( $\displaystyle \sum_{c_i\in \textbf{C}}^{ } |c_i|-2))+ |\textbf{C}|  $
			
		\end{algorithmic}
		\caption{The cycles with limited energy algorithm}
	\end{algorithm}
	
\vspace{-2mm}		
\section{Performance Evaluation}
\label{sec:Experiments}
\vspace{-1mm}
		\subsection{Power Consumption Models}
		\label{subsec:Power_consumption}
		In this section, we quantify the power consumption by each UAV when it is hovering, traveling and transmitting data.
		\subsubsection{Power consumption during hovering}
		The power consumption in watt by the UAV during hovering can be given by~\cite{wannberg2012quadrotor}:
		\begin{equation}
		P=4\frac{T_{h}^{3/2}}{\sqrt{2QS}} + p
		\end{equation}
		
		where $T_{h}$ is the fourth of the quadcopter total weight in $N$, $Q$ is the density of the air in kg/$m^{3}$, $S$ is the rotor swept area in $m^{2}$ and $p$ is the power consumption of electronics in watt.
		
		\subsubsection{Power consumption during traveling}
		The power consumption in kW by the UAV during traveling can be given by~\cite{d2014guest}:
		\begin{equation}
		P=\frac{(m_{p}+m_{v})v}{370\eta r} + p
		\end{equation}
		where $m_{p}$ is the payload mass in kg, $m_{v}$ is the vehicle mass in kg, $r$ is the lift-to-drag ratio (equals 3 for vehicle that is capable of vertical takeoff and landing), $\eta$ is the power transfer efficiency for motor and propeller, $p$ is the power consumption of electronics in kW and $v$ is the velocity in km/h.
		
		\subsubsection{Power consumption during data transmission}
		The power consumption in dB by the UAV during data transmission can be given by~\cite{mozaffari2015drone}:
		\begin{equation}
		P_{t}(dB)=P_{r}(dB)+ \bar{L}(R,h)
		\end{equation}
		\begin{equation}
		\bar{L}(R,h)=P(LOS)\times L_{LOS}+P(NLOS)\times L_{NLOS}
     	\end{equation}
     	\begin{equation}
     	P(LOS)=\dfrac{1}{1+\alpha            .exp(-\beta[\frac{180}{\pi}\theta-\alpha])}
     	\end{equation}
     	\begin{equation}
     	L_{LOS}(dB)=20log(\dfrac{4\pi f_{c}d}{c})+\xi_{LOS}
     	\end{equation}
     	\begin{equation}
     	L_{NLOS}(dB)=20log(\dfrac{4\pi f_{c}d}{c})+\xi_{NLOS}
     	\end{equation}
     	
		In equation (5), $P_{t}$ is the transmit power, $P_{r}$ is the required received power to achieve a SNR greater than threshold $\gamma_{th}$, $\bar{L}(R,h)$ is the average path loss as a function of the altitude $h$ and coverage radius $R$.
		
		In equation (6), $P(LOS)$ is the probability of having line of sight (LOS) connection at an evaluation angle of $\theta$, $P(NLOS)$ is the probability of having non LOS connection and equal (1-$P(LOS)$), $L_{LOS}$ and $L_{NLOS}$ are the average path loss for LOS and NLOS paths.
		
		In equations (7), (8) and (9), $\alpha$ and $\beta$ are constant values which depend on the environment, $f_{c}$ is the carrier frequency, $d$ is the distance between the UAV and the user, $c$ is the speed of the light , $\xi_{LOS}$ and  $\xi_{NLOS}$ are the average additional loss which depend on the environment.
		
		Actually, the power consumed by the UAV during
		data transmission and reception is much smaller than the power consumed during the UAV hovering or traveling~\cite{guptasurvey}. In this paper, we assume that the power wasted during data transmission is constant.
		\vspace{-3mm}
		\subsection{Simulation Setup}
		\label{subsec:simulation_results}
		Given a geographical area $G$ , the number of the subareas that we need to cover and the density of the users, the question here is how to find the optimal boundaries of the subareas that to be covered by the UAVs. To answer this question, the authors of~\cite{mozaffari2016optimal} used transport theory to find the optimal boundaries of the subareas. Unfortunately, this approach needs to solve ${N \choose 2}$ non-linear equations at each iteration, where $N$ is the number of subareas.
		In this paper, we divide the geographical area uniformly and apply the SM and CLE algorithm to find the minimum number of additional UAVs that provides the continuous coverage. We study the effect of the UAV energy capacity, the grid size of the geographical area, the charging time and the traveling time on the number of the additional UAVs.
		Table II lists the parameters used in the numerical analysis~\cite{sentinel}.
			\begin{table}[!h]
				\scriptsize
				\renewcommand{\arraystretch}{1.3}
				\caption{Parameters in numerical analysis}
				\label{table}
				\centering
				\begin{tabular}{|c|c|}
					\hline
					UAV energy capacity & 0.88kW.h \\
					\hline
					Power consumption by the electronics&0.15kW\\
					\hline
					Grid size & 4x4\\
					\hline
					Area of the graph & 1kmx1km\\
					\hline
					Traveling time through edge & 2.5 min\\
					\hline
					Charging station location (x,y)& (0,0)\\
					\hline
					Charging time & 5 min\\
					\hline
					UAV weight with battery & 8.5 k.g\\
					\hline
					Maximum payload  weight & 2 k.g\\
					\hline
					MAX forward speed & 12 m/s\\
					\hline
				\end{tabular}
			\end{table}
		
		\begin{figure}[!h]
			\centering
			\includegraphics[scale=0.43]{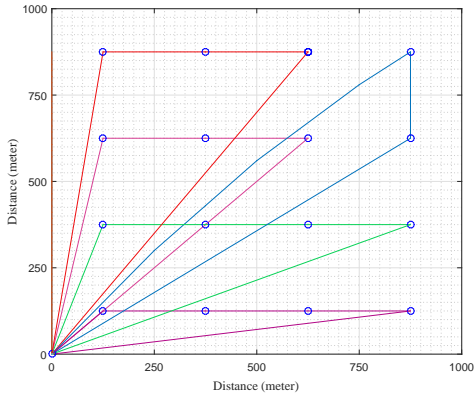}
			\caption{The cycles that cover the subareas using the CLE algorithm}
			\label{fig:cycles}
		\end{figure}
		\begin{figure}[!h]
			\centering
			\includegraphics[scale=0.43]{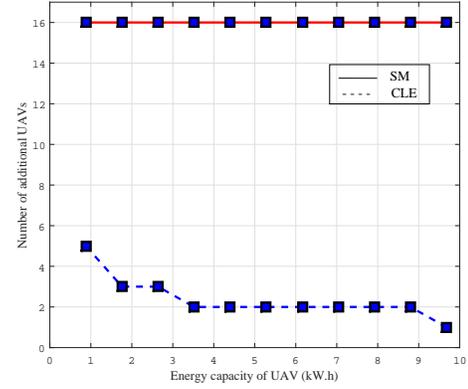}
			\caption{Energy capacity vs. The number of addtional UAVs}
			\label{fig:capacity}
		\end{figure}
			\begin{figure}[!h]
				\centering
				\includegraphics[scale=0.43]{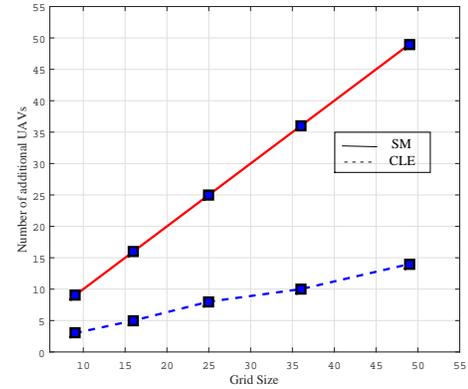}
				\caption{Grid size vs. The number of additional UAVs}
\label{fig:Gridsize}				
			\end{figure}
			
			In Figure~\ref{fig:cycles}, we uniformly divide the geographical area into 16 subareas and apply the CLE algorithm to find the cycles with minimum number of  additional UAVs. From the figure, we notice that 5 cycles are needed to cover all subareas with 5 additional UAVs. Also, we note that the paths of the cycles are intersected in many locations. To avoid the collisions between the UAVs, we operate the paths (cycles) at different altitudes with small altitude differences.
			
			In Figure~\ref{fig:capacity}, we study the effect of the UAV energy capacity on the number of the  additional UAVs needed to cover the subareas. Actually, when increasing the energy capacity of the UAV and apply SM, the number of additional UAVs needed will not change because each subarea is covered by one cycle and two UAVs. When increasing the energy capacity of the
			UAVs, only the coverage time of each UAV  increases. on the other hand, increasing the energy capacity of each UAV results in minimizing the number of additional UAVs that needed using CLE. This is because increasing the energy capacity of each UAV gives the UAV a chance to visit and to cover more subareas, which minimizes the number of the cycles that are needed to cover the subareas.

			\begin{figure}[!t]
				\centering
				\includegraphics[scale=0.43]{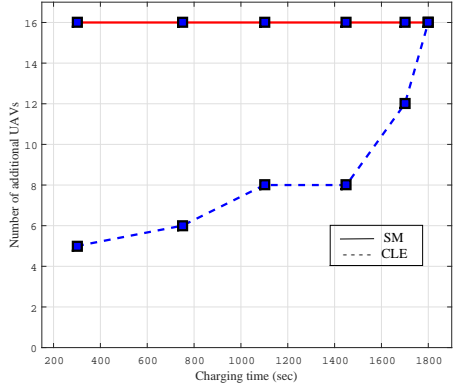}
				\caption{Charging time vs. The number of additional UAVs}
				\label{fig:Charging}
			\end{figure}	 	
			\begin{figure}[!h]
				\centering
				\includegraphics[scale=0.43]{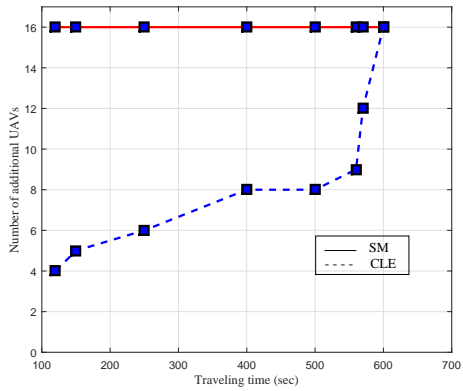}
				\caption{Traveling time vs. The number of additional UAVs}
				\label{fig:Traveling}
			\end{figure}
			In Figure~\ref{fig:Gridsize}, the slope of the line produced by SM is greater than the curve of CLE. When applying SM, the number of additional UAVs increases linearly with the grid size. This is because the number of additional UAVs equals the grid size. Also, when applying the CLE, the number of additional UAVs increases with the grid size. This is because more cycles are needed to cover more subareas and each cycle will need one additional UAV.
			
			In Figure~\ref{fig:Charging}, we study the effect of the charging time on the number of additional UAVs needed. Changing the charging time will	not affect the number of additional UAVs needed when applying SM. This is because the coverage time of each UAV will cover the time that the UAV needs to return to the charging station to recharge the battery and to visit the subarea again. On the other hand, when applying CLE, it will be a critical issue (see Theorem~\ref{theorem_two}). Actually, charging the battery of the UAV takes long time. For this reason, each UAV has a replacement battery~\cite{sentinel}. In this paper, we assume  the time needed to replace the battery for each UAV is 5 minutes.
			
			In Figure~\ref{fig:Traveling}, we study the effect of the traveling time on the number of additional UAVs. Changing the traveling time will not affect the number of additional UAVs when applying SM. On the other hand, it will be a critical issue to choose the appropriate traveling time when applying CLE. When increasing the traveling time, the wasted energy during traveling will increase and the coverage time will decrease. Hence, the chance to visit other subareas will decrease.

\section{Conclusion}
\label{sec:Conclusion}
In this paper, we study the continuous coverage problem with minimum number of replacement UAVs and prove that it is NP complete. We design an efficient algorithm to solve it, the cycles with limited energy algorithm. The proposed algorithm covers the $N$ subareas by cycles, in which each cycle needs one additional UAV to ensure continuous coverage. We showed that the energy capacity of the UAVs, the number of subareas in the affected area, and the UAV charging and traveling times will all impact the required number of UAVs. Our simulation results showed that applying the cycles with limited energy algorithm, can efficiently reduce the number of additional UAVs needed relative to the straightforward method. As future work, we will study the continuous coverage problem using UAVs with non-uniform energy capacities and the use of green energy.
\section*{Acknowledgment}
The work of Jacob Chakareski has been partially supported by the NSF under award CCF-1528030.

\bibliographystyle{IEEEtran}
\bibliography{UAV_conference}

\begin{thebibliography}{10}
\providecommand{\url}[1]{#1}
\csname url@samestyle\endcsname
\providecommand{\newblock}{\relax}
\providecommand{\bibinfo}[2]{#2}
\providecommand{\BIBentrySTDinterwordspacing}{\spaceskip=0pt\relax}
\providecommand{\BIBentryALTinterwordstretchfactor}{4}
\providecommand{\BIBentryALTinterwordspacing}{\spaceskip=\fontdimen2\font plus
\BIBentryALTinterwordstretchfactor\fontdimen3\font minus
  \fontdimen4\font\relax}
\providecommand{\BIBforeignlanguage}[2]{{%
\expandafter\ifx\csname l@#1\endcsname\relax
\typeout{** WARNING: IEEEtran.bst: No hyphenation pattern has been}%
\typeout{** loaded for the language `#1'. Using the pattern for}%
\typeout{** the default language instead.}%
\else
\language=\csname l@#1\endcsname
\fi
#2}}
\providecommand{\BIBdecl}{\relax}
\BIBdecl

\bibitem{bupe2015relief}
P.~Bupe, R.~Haddad, and F.~Rios-Gutierrez, ``Relief and emergency communication
  network based on an autonomous decentralized uav clustering network,'' in
  \emph{SoutheastCon 2015}.\hskip 1em plus 0.5em minus 0.4em\relax IEEE, 2015,
  pp. 1--8.

\bibitem{miller2006hurricane}
R.~Miller, ``Hurricane katrina: Communications \& infrastructure impacts,''
  DTIC Document, Tech. Rep., 2006.

\bibitem{townsend2006federal}
F.~F. Townsend \emph{et~al.}, ``The federal response to hurricane katrina:
  Lessons learned,'' \emph{Washington, DC: The White House}, 2006.

\bibitem{noam2004world}
E.~M. Noam, ``What the world trade center attack has shown us about our
  communications networks,'' \emph{Global Economy and Digital Society.
  Amsterdam: Elsevier Science}, pp. 375--378, 2004.

\bibitem{dalmasso2012wimax}
I.~Dalmasso, I.~Galletti, R.~Giuliano, and F.~Mazzenga, ``Wimax networks for
  emergency management based on uavs,'' in \emph{IEEE First AESS European
  Conference on Satellite Telecommunications (ESTEL)}.\hskip 1em plus 0.5em
  minus 0.4em\relax IEEE, 2012, pp. 1--6.

\bibitem{ha2010uav}
T.~Ha, ``The uav continuous coverage problem,'' DTIC Document, Tech. Rep.,
  2010.

\bibitem{wu2014collaborative}
J.~Wu, ``Collaborative mobile charging and coverage,'' \emph{Journal of
  Computer Science and Technology}, vol.~29, no.~4, pp. 550--561, 2014.

\bibitem{zhang2015collaborative}
S.~Zhang, J.~Wu, and S.~Lu, ``Collaborative mobile charging,'' \emph{IEEE
  Transactions on Computers}, vol.~64, no.~3, pp. 654--667, 2015.

\bibitem{pugliese2015modelling}
L.~D.~P. Pugliese, F.~Guerriero, D.~Zorbas, and T.~Razafindralambo, ``Modelling
  the mobile target covering problem using flying drones,'' \emph{Optimization
  Letters}, pp. 1--32, 2015.

\bibitem{mozaffari2015drone}
M.~Mozaffari, W.~Saad, M.~Bennis, and M.~Debbah, ``Drone small cells in the
  clouds: Design, deployment and performance analysis,'' in \emph{IEEE Global
  Communications Conference (GLOBECOM)}, 2015, pp. 1--6.

\bibitem{mozaffari2016optimal}
------, ``Optimal transport theory for power-efficient deployment of unmanned
  aerial vehicles,'' \emph{IEEE International Conference on Communications
  (ICC), Kuala Lumpur, Malaysia,}, 2016.

\bibitem{zorbas2013energy}
D.~Zorbas, T.~Razafindralambo, F.~Guerriero \emph{et~al.}, ``Energy efficient
  mobile target tracking using flying drones,'' \emph{Procedia Computer
  Science}, vol.~19, pp. 80--87, 2013.

\bibitem{ceran2015optimal}
E.~T. Ceran, T.~Erkilic, E.~Uysal-Biyikoglu, T.~Girici, and K.~Leblebicioglu,
  ``Optimal energy allocation policies for a high altitude flying wireless
  access point,'' \emph{Transactions on Emerging Telecommunications
  Technologies}, 2016.

\bibitem{guptasurvey}
L.~Gupta, R.~Jain, and G.~Vaszkun, ``Survey of important issues in uav
  communication networks,'' \emph{IEEE Communications Surveys and Tutorials},
  2015.

\bibitem{korf2002new}
R.~E. Korf, ``A new algorithm for optimal bin packing,'' in \emph{AAAI/IAAI},
  2002, pp. 731--736.

\bibitem{johnson1997traveling}
D.~S. Johnson and L.~A. McGeoch, ``The traveling salesman problem: A case study
  in local optimization,'' \emph{Local search in combinatorial optimization},
  vol.~1, pp. 215--310, 1997.

\bibitem{wannberg2012quadrotor}
M.~Wannberg, ``The quadrotor platform: from a military point of view,'' 2012.

\bibitem{d2014guest}
R.~D'Andrea, ``Guest editorial can drones deliver?'' \emph{IEEE Transactions on
  Automation Science and Engineering}, vol.~11, no.~3, pp. 647--648, 2014.

\bibitem{sentinel}
``Sentinel+ drone, http://www.airbornedrones.co/,'' 2016.

\end{thebibliography}

\end{document}